\documentclass[10pt]{article}
\newif\ifuseappendix
\useappendixtrue
\newif\ifjournal
\journaltrue
\newcommand{\apxref}[1]{\ifuseappendix Appendix~\ref{#1}\else the extended version of this paper~\cite{supplemental}\fi}
\ifjournal
\usepackage{url}
\else
\usepackage[accepted]{icml2016}
\usepackage{times}
\usepackage{natbib}
\fi
\author{Gregory J.~Puleo \and Olgica Milenkovic}
\usepackage[margin=1in]{geometry}
\usepackage{enumerate}
\usepackage{amsthm}
\usepackage{algorithm}
\usepackage{algorithmic}
\usepackage{mathtools}
\usepackage{amssymb}
\usepackage{xcolor}
\usepackage{graphicx}
\title{Correlation Clustering and Biclustering with Locally Bounded Errors}
\renewcommand{\subset}{\subseteq}
\DeclareMathOperator{\errvec}{err}
\newcommand{\floor}[1]{\lfloor{#1}\rfloor}
\newcommand{\ceil}[1]{\lceil{#1}\rceil}

\newcommand{\st}{\colon\,}
\newcommand{\cee}{\mathcal{C}}
\newcommand{\sey}{\mathcal{S}}

\newcommand{\ints}{\mathbb{Z}}

\newcommand{\pints}{\ints^+}

\newcommand{\mindis}{\textsc{MinDisagree}}
\newcommand{\maxagg}{\textsc{MaxAgree}}
\newcommand{\reals}{\mathbb{R}}
\newcommand{\after}{\circ}

\newcommand{\bad}[1]{\textcolor{red}{\emph{#1}}}
\newcommand{\citethis}[1]{\bad{[CITE THIS]}}
\newcommand{\caze}[2]{\textbf{Case {#1}:} \textit{#2}}
\newcommand{\sizeof}[1]{\left\lvert{#1}\right\rvert}
\newtheorem{assumption}{Assumption}

\newtheorem{proposition}{Proposition}

\newtheorem{observation}[proposition]{Observation}
\newtheorem{lemma}[proposition]{Lemma}
\newtheorem{theorem}[proposition]{Theorem}
\newtheorem{corollary}[proposition]{Corollary}
\theoremstyle{definition}
\newtheorem{definition}[proposition]{Definition}
\theoremstyle{remark}

\newcommand{\LP}{\mathsf{L}}

\begin{document}
\maketitle
\begin{abstract}
  We consider a generalized version of the correlation clustering
  problem, defined as follows. Given a complete graph $G$ whose edges
  are labeled with $+$ or $-$, we wish to partition the graph into
  clusters while trying to avoid errors: $+$ edges between clusters or
  $-$ edges within clusters. Classically, one seeks to minimize the
  total number of such errors. We introduce a new framework that
  allows the objective to be a more general function of the number of
  errors at each vertex (for example, we may wish to minimize the
  number of errors at the \emph{worst} vertex) and provide a rounding
  algorithm which converts ``fractional clusterings'' into discrete
  clusterings while causing only a constant-factor blowup in the
  number of errors at each vertex. This rounding algorithm yields
  constant-factor approximation algorithms for the discrete problem
  under a wide variety of objective functions.
\end{abstract}
\section{Introduction}
Correlation clustering is a clustering model first introduced by
Bansal, Blum, and Chawla~\cite{BBC1,BBC2}.  The basic form of the
model is as follows. We are given a collection of objects and, for
some pairs of objects, we are given a judgment of whether the objects
are \emph{similar} or \emph{dissimilar}. This information is
represented as a labeled graph, with edges labeled $+$ or $-$
according to whether the endpoints are similar or dissimilar. Our goal
is to cluster the graph so that $+$ edges tend to be within clusters
and $-$ edges tend to go across clusters.  The number of clusters is
not specified in advance; determining the optimal number of clusters
is instead part of the optimization problem.

Given a solution clustering, an \emph{error} is a $+$ edge whose
endpoints lie in different clusters or a $-$ edge whose endpoints lie
in the same cluster. In the original formulation of the correlation
clustering, the goal is to minimize the total number of errors; this
formulation of the optimization problem is called \mindis. Finding an
exact optimal solution is NP-hard even when the input graph is
complete \cite{BBC1,BBC2}. Furthermore, if the input graph is allowed
to be arbitrary, the best known approximation ratio is $O(\log n)$,
obtained by~\cite{CGW1,CGW2,DEFI1}. Assuming the Unique Games
Conjecture of Khot~\cite{khotUGC}, no constant-factor approximation
for \mindis{} on arbitrary graphs is possible; this follows from
the results of \cite{UGC1, UGC2} concerning the minimum multicut problem
and the connection between correlation clustering and minimum multicut
described in \cite{CGW1, CGW2, DEFI1}.

Since theoretical barriers appear to preclude constant-factor
approximations on arbitrary graphs, much research has focused on
special graph classes such as complete graphs and complete bipartite
graphs, which are the graph classes we consider here.  Ailon,
Charikar, and Newman~\cite{ACN1,ACN2} gave a very simple randomized
$3$-approximation algorithm for \mindis{} on complete graphs. This
algorithm was derandomized by van~Zuylen and Williamson~\cite{VZW},
and a parallel version of the algorithm was studied by Pan,
Papailiopoulos, Recht, Ramchandran, and Jordan~\cite{jordan}.  More
recently, a $2.06$-approximation algorithm was announced by Chawla,
Makarychev, Schramm and Yaroslavtsev~\cite{near-optimal}.  Similar
results have been obtained for complete bipartite graphs. The first
constant approximation algorithm for correlation clustering on
complete bipartite graphs was described by
Amit~\cite{amit2004bicluster}, who gave an $11$-approximation
algorithm. This ratio was improved by Ailon, Avigdor-Elgrabli, Liberty
and van Zuylen~\cite{ailon2012improved}, who obtained a
$4$-approximation algorithm. Chawla, Makarychev, Schramm and
Yaroslavtsev~\cite{near-optimal} announced a $3$-approximation
algorithm for correlation clustering on complete $k$-partite graphs,
for arbitrary $k$, which includes the complete bipartite case.
Bipartite clustering has also been studied, outside the
correlation-clustering context, by Lim, Chen, and
Xu~\cite{lim2015convex}.

We depart from the classical correlation-clustering literature by
considering a broader class of objective functions which also cater to
the need of many community-detection applications in machine learning,
social sciences, recommender systems and
bioinformatics~\cite{cheng2000biclustering,symeonidis2007nearest,kriegel2009clustering}.
The technical details of this class of functions can be found in
Section~\ref{sec:formal}.  As a representative example of this class, we
introduce \emph{minimax correlation clustering}.

In minimax clustering, rather than seeking to minimize the
\emph{total} number of errors, we instead seek to minimize the number
of errors at the \emph{worst-off vertex} in the clustering. Put more
formally, if for a given clustering each vertex $v$ has $y_v$ incident
edges that are errors, then we wish to find a clustering that
minimizes $\max_v y_v$.

Minimax clustering, like classical correlation clustering, is NP-hard
on complete graphs, as we prove in \apxref{apx:npc}. To design
approximation algorithms for minimax clustering, it is necessary to
bound the growth of errors \emph{locally} at each vertex when we round
from a fractional clustering to a discrete clustering; this introduces
new difficulties in the design and analysis of our rounding
algorithm. These new technical difficulties cause the algorithm of
\cite{ACN1,ACN2} to fail in the minimax context, and there is no
obvious way to adapt that algorithm to this new context; this
phenomenon is explored further in \apxref{apx:minimax}.

Minimax correlation clustering on graphs is relevant in detecting
communities, such as gene, social network, or voter communities, in
which no \emph{antagonists} are allowed. Here, an antagonist refers to
an entity that has properties inconsistent with a large number of
members of the community. Alternatively, one may view the minimax
constraint as enabling individual vertex quality control within the
clusters, which is relevant in biclustering applications such as
collaborative filtering for recommender systems, where minimum quality
recommendations have to be ensured for each user in a given category.
As an illustrative example, one may view a complete bipartite graph as
a preference model in which nodes on the left represent viewers and
nodes on the right represent movies. A positive edge between a user
and a movie indicates that the viewer likes the movie, while a
negative edge indicates that they do not like or have not seen the
movie. We may be interested in finding communities of viewers for the
purpose of providing them with joint recommendations. Using a minimax
objective function here allows us to provide a uniform quality of
recommendations, as we seek to minimize the number of errors for the
user who suffers the most errors.

A minimax objective function for a graph partitioning problem
different from correlation clustering was previously studied by
\cite{bansal-minmax}.  In that paper, the problem under consideration
was to split a graph into $k$ roughly-equal-sized parts, minimizing
the total number of edges leaving any part. Thus, the minimum in
\cite{bansal-minmax} is being taken over the parts of the solution,
rather than minimizing over vertices as we do here.

Another idea slightly similar to minimax
clustering has previously appeared in the literature on
fixed-parameter tractability of the \textsc{Cluster Editing} problem,
which is an equivalent formulation of Correlation Clustering.  In
particular, Komusiewicz and Uhlmann~\cite{kom} proved that the
following problem is fixed-parameter tractable for the combined
parameter $(d,t)$:
\begin{quote}
  \textbf{$(d,t)$-Constrained-Cluster Editing}\\
  \textbf{Input:} A labeled complete graph $G$, a function $\tau : V(G) \to \{0, \ldots, t\}$, and nonnegative integers $d$ and $k$.\\
  \textbf{Question:} Does $G$ admit a clustering into at most $d$ clusters with at most $k$ errors
  such that every vertex $v$ is incident to at most $\tau(v)$ errors?
\end{quote}
(Here, we have translated their original formulation into the language
of correlation clustering.)  Komusiewicz and Uhlmann also obtained
several NP-hardness results related to this formulation of the
problem. While their work involves a notion of local errors for
correlation clustering, their results are primarily focused on
fixed-parameter tractability, rather than approximation algorithms,
and are therefore largely orthogonal to the results of this paper.

The contributions of this paper are organized as follows. In
Section~\ref{sec:formal}, we introduce and formally express our
framework for the generalized version of correlation clustering, which
includes both classical clustering and minimax clustering as special
cases. In Section~\ref{sec:algcc}, we give a rounding algorithm which
allows the development of constant-factor approximation algorithms for
the generalized clustering problem. In Section~\ref{sec:algbc}, we
give a version of this rounding algorithm for complete bipartite
graphs.

\ifuseappendix In Appendix~\ref{apx:minimax}, we discuss minimax
clustering in more detail, and show that algorithms similar to the
Ailon--Charikar--Newman algorithm fail in the minimax context.  In
Appendix~\ref{apx:maxagg} we discuss the approximation properties of
the \maxagg{} formulation of minimax clustering, where the objective
is to maximize the number of correct edges, rather than minimize the
number of incorrect edges, at the worst vertex.  In
Appendix~\ref{apx:npc} and Appendix~\ref{apx:bip-npc} we prove that
the minimax correlation clustering problem is NP-hard on complete
graphs and complete bipartite graphs,
respectively. Appendix~\ref{apx:details} contains technical details
for various proofs.  \fi
\section{Framework and Formal Definitions}\label{sec:formal}
In this section, we formally set up the framework we will use
for our broad class of correlation-clustering objective functions.
\begin{definition}
  Let $G$ be an edge-labeled graph. A \emph{discrete clustering} (or
  just a \emph{clustering}) of $G$ is a partition of $V(G)$.  A
  \emph{fractional clustering} of $G$ is a vector $x$ indexed by
  ${V(G) \choose 2}$ such that $x_{uv} \in [0,1]$ for all $uv \in
  {V(G) \choose 2}$ and such that $x_{vz} \leq x_{vw} + x_{wz}$ for
  all distinct $v,w,z \in V(G)$.
\end{definition}
If $x$ is a fractional clustering, we can view $x_{uv}$ as a
``distance'' from $u$ to $v$; the constraints $x_{vz} \leq x_{vw} +
x_{wz}$ are therefore referred to as \emph{triangle inequality}
constraints. We also adopt the convention that $x_{uu}=0$ for all $u$.

In the special case where all coordinates of $x$ are $0$ or $1$, the
triangle inequality constraints guarantee that the relation defined by
$u \sim v$ iff $x_{uv} = 0$ is an equivalence relation. Such a vector
$x$ can therefore naturally be viewed as a discrete clustering, where
the clusters are the equivalence classes under $\sim$. By viewing a
discrete clustering as a fractional clustering with integer
coordinates, we see that fractional clusterings are a continuous
relaxation of discrete clusterings, which justifies the name. This
gives a natural notion of the total weight of errors at a given
vertex.
\begin{definition}
  Let $G$ be an edge-labeled complete graph, and let $x$ be a fractional clustering of $G$.
  The \emph{error vector} of $x$ with respect to $G$, written
  $\errvec(x)$, is a real vector indexed by $V(G)$ whose coordinates
  are defined by
  \[ \errvec(x)_v = \sum_{w \in N^+(v)}x_{vw} + \sum_{w \in N^-(v)}(1 - x_{vw}). \]
  If $\cee$ is a clustering of $G$ and $x^{\cee}$ is the natural
  associated fractional clustering, we define $\errvec(\cee)$ as
  $\errvec(x^{\cee})$.
\end{definition}
We are now prepared to formally state the optimization problem we wish to solve.
Let $\reals^n_{\geq 0}$ denote the set of vectors in $\reals^n$ with all coordinates
nonnegative.
Our problem is parameterized by a function $f : \reals^{n}_{\geq 0} \to \reals$.
\begin{quote}
  \textbf{$f$-Correlation Clustering}\\
  \textbf{Input:} A labeled graph $G$.\\
  \textbf{Output:} A clustering $\cee$ of $G$.\\
  \textbf{Objective:} Minimize $f(\errvec(\cee))$.
\end{quote}
In order to approximate $f$-Correlation Clustering, we introduce a relaxed
version of the problem.
\begin{quote}
  \textbf{Fractional $f$-Correlation Clustering}\\
  \textbf{Input:} A labeled graph $G$.\\
  \textbf{Output:} A fractional clustering $x$ of $G$.\\
  \textbf{Objective:} Minimize $f(\errvec(x))$.
\end{quote}
If $f$ is convex on $\reals_{\geq 0}^n$, then using standard
techniques from convex optimization \cite{boyd-cvxopt}, the Fractional
$f$-Correlation Clustering problem can be approximately solved in
polynomial time, as the composite function $f \after \errvec$ is
convex and the constraints defining a fractional clustering are linear
inequalities in the variables $x_e$. When $G$ is a complete graph, we
then employ a rounding algorithm based on the algorithm of Charikar,
Guruswami, and Wirth~\cite{CGW1,CGW2} to transform the fractional
clustering into a discrete clustering. Under rather modest conditions
on $f$, we are able to obtain a constant-factor bound on the error
growth, that is, we can produce a clustering $\cee$ such that
$f(\errvec(\cee)) \leq c f(\errvec(x))$, where $c$ is a constant not
depending on $f$ or $x$. In particular, we require the following
assumptions on $f$.
\begin{assumption}\label{fgood}
  We assume that $f : \reals^{n}_{\geq 0} \to \reals$ has the following properties.
  \begin{enumerate}[(1)]
  \item $f(cy) \leq c f(y)$ for all $c \geq 0$ and all $y \in \reals^{n}$, and
  \item If $y,z \in \reals^n_{\geq 0}$ are vectors with $y_i \leq z_i$ for all $i$, 
    then $f(y) \leq f(z)$.
  \end{enumerate}
\end{assumption}
Under Assumption~\ref{fgood}, the claim that $f(\errvec(\cee)) \leq c
f(\errvec(x))$ follows if we can show that $\errvec(\cee)_v \leq c
\errvec(x)_v$ for every vertex $v \in V(G)$. This is the property we
prove for our rounding algorithms.

We will slightly abuse terminology by referring to the constant $c$ as
an \emph{approximation ratio} for the rounding algorithm; this
notation is motivated by the fact that when $f$ is linear, the
Fractional $f$-Correlation Clustering problem can be solved exactly in
polynomial time, and applying a rounding algorithm with constant $c$
to the fractional solution yields a $c$-approximation algorithm to the
(discrete) $f$-Correlation Clustering problem. In contrast, when $f$
is nonlinear, we may only be able to obtain a
$(1+\epsilon)$-approximation for the Fractional $f$-Correlation
Clustering problem, in which case applying the rounding algorithm
yields a $c(1+\epsilon)$-approximation algorithm for the discrete
problem.

A natural class of convex objective functions obeying
Assumption~\ref{fgood} is the class of $\ell^p$ norms. For all $p \geq
1$, the $\ell^p$-norm on $\reals^n$ is defined by
\[ \ell^p(x) = \left(\sum_{i=1}^n \sizeof{x_i}^p\right)^{1/p}. \]
As $p$ grows larger, the $\ell^p$-norm puts more emphasis on the coordinates with
larger absolute value. This justifies that definition of the $\ell^{\infty}$-norm as
\[ \ell^{\infty}(x) = \max\{x_1, \ldots, x_n\}. \]
Classical correlation clustering is the case of $f$-Correlation Clustering
where $f(x) = \frac{1}{n}\ell^1(x)$, while minimax correlation clustering is
the case of $f$-Correlation Clustering where $f(x) = \ell^{\infty}(x)$. 

Our emphasis on convex $f$ is due to the fact that convex programming
techniques allow the Fractional $f$-Correlation Clustering problem to
be approximately solved in polynomial time when $f$ is convex. However, the correctness
of our rounding algorithm does not depend on the convexity of $f$,
only on the properties listed in Assumption~\ref{fgood}. If $f$ is
nonconvex and obeys Assumption~\ref{fgood}, and we produce a ``good''
fractional clustering $x$ by some means, then our algorithm still
produces a discrete clustering $\cee$ with $f(\errvec(\cee)) \leq
cf(\errvec(x))$.

\section{A Rounding Algorithm for Complete Graphs}\label{sec:algcc}
We now describe a rounding algorithm to transform an arbitrary
fractional clustering $x$ of a labeled complete graph $G$ into a
clustering $\cee$ such that $\errvec(\cee)_v \leq c \errvec(x)_v$ for
all $v \in V(G)$.

Our rounding algorithm is based on the algorithm of Charikar,
Guruswami, and Wirth~\cite{CGW1,CGW2} and is shown in
Algorithm~\ref{alg:round}. The main difference between
Algorithm~\ref{alg:round} and the algorithm of \cite{CGW1,CGW2} is the
new strategy of choosing a pivot vertex that maximizes
$\sizeof{T^*_u}$; in \cite{CGW1,CGW2}, the pivot vertex is chosen
arbitrarily. Furthermore, the algorithm of \cite{CGW1,CGW2} always
uses $\alpha=1/2$ as a cutoff for forming ``candidate clusters'',
while we express $\alpha$ as a parameter which we later choose in
order to optimize the approximation ratio.
\begin{algorithm*}
  \caption{Round fractional clustering $x$ to obtain a discrete clustering, using threshold parameters $\alpha, \gamma$ with $0 < \gamma < \alpha < 1/2$.}
  \label{alg:round}
  \begin{algorithmic}
    \STATE{Let $S = V(G)$.}
    \WHILE{$S \neq \emptyset$}
    \STATE{For each $u \in S$, let $T_u = \{w \in S-\{u\} \st x_{uw} \leq \alpha\}$ and let $T^*_u = \{w \in S-\{u\} \st x_{uw} \leq \gamma\}$.}
    \STATE{Choose a pivot vertex $u \in S$ that maximizes $\sizeof{T^*_u}$.}
    \STATE{Let $T = T_u$.}
    \IF{$\sum_{w \in T}x_{uw} \geq \alpha\sizeof{T}/2$}
    \STATE{Output the cluster $\{u\}$.} \COMMENT{Type $1$ cluster}
    \STATE{Let $S = S-\{u\}$.}
    \ELSE
    \STATE{Output the cluster $\{u\} \cup T$.} \COMMENT{Type $2$ cluster}
    \STATE{Let $S = S - (\{u\} \cup T)$.}
    \ENDIF
    \ENDWHILE      
  \end{algorithmic}
\end{algorithm*}

Under the classical objective function, an optimal fractional
clustering is the solution to a linear program, which motivates the
following notation for the more general case.
\begin{definition}
  If $uv$ is an edge of a labeled graph $G$, we define the
  \emph{LP-cost} of $uv$ relative to a fractional clustering
  $x$ to be $x_{uv}$ if $uv \in E^+$, and $1-x_{uv}$ if $uv \in E^-$.
  Likewise, the \emph{cluster-cost} of an edge $uv$ is $1$ if $uv$ is an
  error in the clustering produced by Algorithm~\ref{alg:round}, and $0$
  otherwise.
\end{definition}
Our general strategy for obtaining the constant-factor error bound for
Algorithm~\ref{alg:round} is similar to that of \cite{CGW1,CGW2}. Each
time a cluster is output, we pay for the cluster-cost of the errors
incurred by ``charging'' the cost of these errors to the LP-costs of
the fractional clustering.  The main difference between our proof and
the proof of \cite{CGW1,CGW2} is that we must pay for errors
\emph{locally}: for each vertex $v$, we must pay for \emph{all}
clustering errors incident to $v$ by charging to the LP cost incident
to $v$. In particular, every clustering error must now be paid for at
\emph{each} of its endpoints, while in \cite{CGW1,CGW2}, it was enough
to pay for each clustering error at \emph{one} of its endpoints. For
edges which cross between a cluster and its complement, this requires
a different analysis at each endpoint, a difficulty which was not
present in \cite{CGW1,CGW2}. Our proof emphasizes the solutions to
these new technical problems; the parts of the proof that are
technically nontrivial but follow earlier work are omitted due to
space constraints but can be found in \apxref{apx:details}.

\begin{observation}\label{obs}
  Let $x$ be a fractional clustering of a graph $G$, and let $w,z \in V(G)$.
  For any vertex $u$, we have $x_{wz} \geq x_{uz} - x_{uw}$ and
  $1 - x_{wz} \geq 1 - x_{uz} - x_{uw}$.
\end{observation}
\begin{theorem}\label{thm:complete-appx}
  Let $G$ be a labeled complete graph, let $\alpha$ and $\gamma$ be
  parameters with $0 < \gamma < \alpha < 1/2$, and let $x$ be any
  fractional clustering of $G$. If $\cee$ is the clustering produced
  by Algorithm~\ref{alg:round} with the given input, then for all
  $v \in V(G)$ we have $\errvec(\cee)_v \leq c\errvec(x)_v$, where $c$
  is a constant depending only on $\alpha$ and $\gamma$.
\end{theorem}
\begin{proof}
  Let $k_1, k_2, k_3$ be constants to be determined, with
  $1/2 < k_1 < 1$ and $0 < 2k_2 \leq k_3 < 1/2$.  Also assume
  that $k_1\alpha > \gamma$ and that $k_2\alpha \leq 1-2\alpha$.

  To prove the approximation ratio, we consider the cluster-costs
  incurred as each cluster is output, splitting into cases according
  to the type of cluster. In our analysis, as the algorithm runs, we
  will mark certain vertices as ``safe'', representing the fact that
  some possible future clustering costs have been paid for in
  advance. Initially, no vertex is marked as safe.

  \caze{1}{A Type 1 cluster is output.} Let
  $X = S \cap N^+(u)$, with $S$ as in Algorithm~\ref{alg:round}.  The
  new cluster-cost incurred at $u$ is $\sizeof{X}$, and for each
  $v \in X$, a new cluster-cost of $1$ is incurred at $v$.

  First we pay for the new cluster cost incurred at $u$. For each edge
  $uv$ with $v \in T$, we have $x_{uv} \leq \alpha$ and so $1 - x_{uv}
  \geq 1-\alpha \geq x_{uv}$. Thus, the total LP cost of edges $uv$
  with $v \in T$ is at least $\sum_{v \in T}x_{uv}$, which is at least
  $\alpha\sizeof{T}/2$ since $\{u\}$ is output as a Type 1 cluster. Thus,
  charging each edge $uv$ with $v \in T$ a total of $2/\alpha$ times
  its LP-cost pays for the cluster-cost of any positive edges from $u$
  to $T$.  On the other hand, if $uv$ is a positive edge with $v \in S
  - T$, then since $v \notin T$, we have $x_{uv} \geq \alpha$.  Hence,
  the LP-cost of $uv$ is at least $\alpha$, and charging $1/\alpha$
  times the LP-cost of $uv$ pays for the cluster-cost of this edge.

  Now let $v \in X$; we must pay for the new cluster cost at $v$. If
  $x_{uv} \geq k_2\alpha$, then the edge $uv$ already incurs LP
  cost at least $k_2\alpha$, so the new cost at $v$ is only $1/(k_2\alpha)$ times
  the LP-cost of the edge $uv$. So assume $x_{uv} < k_2\alpha$. In this case,
  we say that $u$ is a \emph{bad pivot} for $v$.

  First suppose that $v$ is not safe (as is initially the case). We
  will make a single charge to the edges incident to $v$ that is large
  enough to pay for both the edge $uv$ and for all possible
  \emph{future} bad pivots, and then we will mark $v$ as safe to
  indicate that we have done this. The basic idea is that if $v$ has
  many possible bad pivots, then since $x_{uv}$ is ``small'', all of
  these possible bad pivots are also close to $u$, thus included in
  $T_u$.  Since $\sum_{w \in T_u}x_{uw} \geq \alpha\sizeof{T_u}/2$,
  there is a large set $B \subset T_u$ of vertices that are
  ``moderately far'' from $u$, and therefore moderately far from $v$.
  The number of these vertices grows with the number of bad pivots, so
  charging all the edges $vz$ for $z \in B$ is sufficient to pay for
  all bad pivots.

  We now make this argument rigorous. Let $P_v$ be the set of
  potential bad pivots for $v$, defined by
  \[ P_v = \{p \in S \st x_{vp} < k_2\alpha\}. \]
  Note that $u \in P_v$. Since $k_2 < 1/4$, we have $x_{up} \leq
  x_{uv} + x_{vp} < \alpha/2$ for all $p \in P_v$; hence $P_v \subset T$.
  Define the vertex set $B$ by
  \[ B = \{z \in T \st x_{uz} > k_3\alpha\}. \]
  Since $x_{uz} \leq \alpha$ for all $z \in T$, we see that
  \[ \sum_{z \in T}x_{uz} \leq k_3\alpha\sizeof{T-B} +
  \alpha\sizeof{B}. \]
  On the other hand, since $\{u\}$ is output as a Type 1 cluster, we
  have \[\sum_{z \in T}x_{uz} \geq \alpha\sizeof{T}/2.\]
  Combining these inequalities and rearranging, we obtain
  $\sizeof{B} \geq (1-2k_3)\sizeof{T-B}$. For each vertex $z \in B$,
  we have $x_{vz} \geq x_{uz} - x_{uv} \geq (k_3-k_2)\alpha$; in particular,
  since $k_3 \geq 2k_2$, we have $x_{vz} \geq k_2\alpha$, so that $z \notin P_v$.
  Hence $\sizeof{T-B} \geq \sizeof{P_v}$, and we have $\sizeof{B} \geq (1-2k_3)\sizeof{P_v}$.

  On the other hand, for $z \in B$ we also have $1 - x_{vz} \geq 1 - x_{uv} - x_{uz} \geq 1-(1+k_2)\alpha$.
  It follows that each edge $vz$ for $z \in B$ has LP-cost at least $\min((k_3-k_2)\alpha,\ 1-(1+k_2)\alpha)$,
  independent of whether $vz$ is positive or negative. It is easy to check that since $\alpha < 1/2$ and $k_3 < 1$, this
  minimum is always achieved by $(k_3-k_2)\alpha$. Therefore, we can pay for the (possible)
  Type-1-cluster cost of all edges $vp$ for $p \in P_v$ by charging each edge $vz$ with $z \in B$
  a total of
  \[ \frac{1}{(1-2k_3)(k_3-k_2)\alpha} \]
  times its LP-cost. We make all these charges when the cluster $\{u\}$ is created and put
  them in a ``bank account'' to pay for later Type-1-cluster costs for $v$. Then we mark
  $v$ as safe. The total charge in the bank account is at least $\sizeof{P_v}$, which is enough
  to pay for all bad pivots for $v$.

  We have just described the case where $u$ is a bad pivot and $v$ is
  not safe. On the other hand, if $u$ is a bad pivot and $v$ is safe,
  then $v$ already has a bank account large enough to pay for all its
  bad pivots, and we simply charge $1$ to the account to pay for the
  edge $uv$.  \medskip

  \caze{2}{A Type 2 cluster $\{u\} \cup T$ is output.} The negative edges
  within $\{u\} \cup T$ are easy to pay for: if $vw$ if a negative edge inside
  $\{u\} \cup T$, then we have $1-x_{vw} \geq 1-x_{uv}-x_{uw} \geq 1-2\alpha$,
  so we can pay for each of these edges by charging a factor of $\frac{1}{1-2\alpha}$
  times its LP-cost.

  Thus, we consider edges joining $\{u\} \cup T$ with $S - (\{u\} \cup
  T)$. We call these edges \emph{cross-edges} for their endpoints.  A
  standard argument (see \apxref{apx:details}) shows that for
  $z \in S - (\{u\} \cup T)$, the total cluster-cost of the
  cross-edges for $z$ is at most $\max\{1/(1-2\alpha), 2/\alpha\}$
  times the LP-cost of those edges, so the vertices outside $\{u\}
  \cup T$ can be dealt with easily.

  However, we also must bound the cluster-cost at vertices inside
  $\{u\} \cup T$. This is where we use the maximality of
  $\sizeof{T^*_u}$.

  Let $w \in \{u\} \cup T$. First consider the positive cross-edges
  $wz$ such that $x_{wz} \geq \gamma$. Any such edge has cluster-cost
  $1$ and already has LP-cost at least $\gamma$, so charging
  $1/\gamma$ times the LP-cost to such an edge pays for its cluster
  cost. Now let $X = \{z \in S - (\{u\} \cup T) \st x_{wz} < \gamma\}$; we
  still must pay for the edges $wz$ with $z \in X$.

  If $x_{uw} \leq k_1\alpha$, which includes the case $u=w$, then for
  all $z \in X$, we have $x_{wz} \geq x_{uz} - x_{uw} \geq \alpha -
  k_1\alpha = (1-k_1)\alpha$. Hence, for any positive edge $wz$ with
  $z \in X$, the LP-cost of $wz$ is at least $(1-k_1)\alpha$, and so
  the cluster cost of the edge $wz$ is at most $1/((1-k_1)\alpha)$
  times the LP cost. Charging this factor to each cross-edge pays for
  the cluster-cost of each cross-edge.

  Now suppose $x_{uw} > k_1\alpha$. Since $k_1\alpha > \gamma$, this
  implies $w \notin T^*_u$. In this case, it is possible that $w$ may
  have many positive neighbors $z \in X$ for which $x_{wz}$ is quite
  small, so we cannot necessarily pay for the cluster-cost of the
  edges joining $w$ and $X$ by using their LP-cost. Instead, we charge
  their cluster-cost to the LP-cost of edges within $T$.

  Observe that $X \subset T^*_w$, and hence $\sizeof{T^*_w} \geq \sizeof{X}$.
  By the maximality of $\sizeof{T^*_u}$, this implies that $\sizeof{T^*_u} \geq \sizeof{X}$. 
  Now for any $v \in T^*_u$, we have the following bounds:
  \begin{align*}
    x_{wv} &\geq x_{uw} - x_{uv} \geq k_1\alpha - \gamma,\\
    1-x_{wv} &\geq 1 - x_{uw} - x_{uv} \geq 1 - \alpha - \gamma.    
  \end{align*}
  Since $\alpha < 1/2$ and $k_1 \leq 1$, we have $k_1\alpha \leq \alpha < 1-\alpha$, so these
  lower bounds imply that each edge $wv$ with $v \in T^*_u$ has LP-cost at least $k_1\alpha-\gamma$, independent
  of whether $wv$ is a positive or negative edge. Thus, the total LP cost of edges joining $w$ to $T^*_u$
  is at least $(k_1\alpha-\gamma)\sizeof{T^*_u}$.

  Since the total cluster-cost of edges joining $w$ and $X$ is at most
  $\sizeof{X}$ and since $\sizeof{T^*_u} \geq \sizeof{X}$, we can pay
  for these edges by charging each edge $wv$ with $v \in T^*_u$ a
  factor of $\frac{1}{k_1\alpha-\gamma}$ times its LP-cost.

  \medskip
  Having paid for all cluster-costs, we now look at the total charge accrued at each vertex.
  Fix any vertex $v$ and an edge $vw$ incident to $v$. We bound the total amount charged to $vw$
  by $v$ in terms of the LP-cost of $vw$. There are three distinct possibilities for the edge $vw$:
  either $vw$ ended inside a cluster, or $v$ was clustered before $w$, or $w$ was clustered before $v$.

  \caze{1}{$vw$ ended within a cluster.} In this case, $v$ may have made the following charges:
  \begin{itemize}
  \item A charge of $\frac{1}{(1-2k_3)(k_3-k_2)\alpha}$ times the LP-cost, to pay for a ``bank account'' for $v$,
  \item A charge of $\frac{1}{1-2\alpha}$ times the LP-cost, to pay for $vw$ itself if $vw$ is a negative edge,
  \item A charge of $\frac{1}{k_1\alpha-\gamma}$ times the LP-cost, to pay for positive edges leaving the $v$-cluster.
  \end{itemize}
  Thus, in this case the total cost charged to $vw$ by $v$ is at most
  $c_1$ times the LP-cost of $vw$, where
  \[
    c_1 = \frac{1}{(1-2k_3)(k_3-k_2)\alpha} + \frac{1}{1-2\alpha} + \frac{1}{k_1\alpha-\gamma}.
    \]

  \caze{2}{$v$ was clustered before $w$.} In this case, $v$ may have made the following charges:
  \begin{itemize}
  \item A charge of $\frac{1}{(1-2k_3)(k_3-k_2)\alpha}$ times the LP-cost, to pay for a ``bank account'' for $v$,
  \item A charge of at most $\frac{2}{\alpha}$ times the LP-cost, to pay for all cross-edges if $v$ was output as a Type 1 cluster,
  \item A charge of at most $\max\left\{\frac{1}{(1-k_1)\alpha},\ \frac{1}{\gamma} \right\}$ times the LP-cost,
  to pay for $vw$ if $v$ was output in a Type 2 cluster.
  \end{itemize}
  Note that $k_1 > 1/2$ implies that $\frac{1}{(1-k_1)\alpha} \geq \frac{2}{\alpha}$, so we may disregard the case
  where $v$ is output as a Type 1 cluster.
  Thus, in this case the total cost charged to $vw$ by $v$ is at most $c_2$ times the LP-cost of $vw$, where
  \[ c_2 = \frac{1}{(1-2k_3)(k_3-k_2)\alpha} + \max\left\{\frac{1}{(1-k_1)\alpha},\ \frac{1}{\gamma}\right\}. \]

  \caze{3}{$w$ was clustered before $v$.} In this case, $v$ may have made the following charges:
  \begin{itemize}
  \item A charge of at most $\frac{1}{(1-2k_3)(k_3-k_2)\alpha}$ times the LP-cost, to pay for a ``bank account'' for $v$,
  \item A charge of at most $\frac{1}{k_2\alpha}$ times the LP-cost, to pay for the cluster-cost of $vw$ if $vw$ is a positive edge
    and $w$ was output as a Type 1 cluster,
  \item A charge of at most
    \[ \max\left\{\frac{1}{1-2\alpha},\ \frac{2}{\alpha}\right\} \]
    times the LP-cost, to pay for $vw$ if $w$ was output in a Type 2 cluster.
  \end{itemize}
  Clearly $vw$ cannot receive both the second and third types of
  charge. Furthermore, since $k_2 \leq 1/4$, we have
  $\frac{1}{k_2\alpha} \geq \frac{2}{\alpha}$. Since $k_2\alpha \leq
  1-2\alpha$, we see that $\frac{1}{k_2\alpha}$ is the largest charge
  that $vw$ could receive from either the second or third type of
  charge.  Thus, in this case the total cost charged to $vw$ by $v$ is
  at most $c_3$ times the LP-cost, where
  \[ c_3 = \frac{1}{(1-2k_3)(k_3-k_2)\alpha} + \frac{1}{k_2\alpha.} \]
  \smallskip

  Thus, the approximation ratio of the algorithm is at most
  $\max\{c_1, c_2, c_3\}$. We wish to choose the various parameters to
  make this ratio as small as possible, subject to the various
  assumptions on the parameters required for the correctness of the
  proof. It seems difficult to obtain an exact solution to this
  optimization problem. Solving the problem numerically, we obtained
  the following values for the parameters:
  \begin{align*}
    \alpha &= 0.465744 & \gamma &=0.0887449 \\
    k_1 &= 0.767566 & k_2 &= 0.117219 & k_3 &= 0.308433.
  \end{align*}
  These parameters yield an approximation ratio of roughly $48$.
\end{proof}

\section{A Rounding Algorithm for One-Sided Biclustering}\label{sec:algbc}
In this section, we consider a version of the $f$-Correlation
Clustering problem on complete bipartite graphs.  Let $G$ be a
complete bipartite graph with edges labeled $+$ and $-$, and let $V_1$
and $V_2$ be its partite sets. We will obtain a rounding algorithm
that transforms any fractional clustering $x$ into a discrete
clustering $\cee$ such that $\errvec(\cee)_v \leq c \errvec(x)_v$ for
all $v \in V_1$.  Our algorithm is shown in
Algorithm~\ref{alg:bip-round}.

Our algorithm does not guarantee any upper bound on $\errvec(\cee)_v$
for $v \in V_2$: as the algorithm treats the sides $V_1$ and $V_2$
asymmetrically, it is difficult to control the per-vertex error at
$V_2$. Nevertheless, an error guarantee for the vertices in $V_1$
suffices for some applications. Our approach is motivated by
applications in recommender systems, where vertices in $V_1$
correspond to users, while vertices in $V_2$ correspond to objects to
be ranked.  In this context, quality of service conditions only need
to be imposed for users, and not for objects.

\begin{algorithm*}
  \caption{Round fractional clustering to obtain a discrete clustering, using
    threshold parameters $\alpha, \gamma$ with
    $\alpha < 1/2$ and $\gamma < \alpha$.}
  \label{alg:bip-round}
  \begin{algorithmic}
    \STATE{Let $S = V(G)$.}
    \WHILE{$V_1 \cap S \neq \emptyset$}
    \STATE{For each $u \in V_1\cap S$, let $T_u = \{w \in S-\{u\} \st x_{uw} \leq \alpha\}$
      and let $T^*_u = \{w \in V_2 \cap S \st x_{uw} \leq \gamma\}$.}
    \STATE{Choose a pivot vertex $u \in V_1 \cap S$ that maximizes $\sizeof{T^*_u}$.}
    \STATE{Let $T = T_u$.}
    \IF{$\sum_{w \in V_2 \cap T}x_{uw} \geq \alpha\sizeof{V_2 \cap T}/2$}
    \STATE{Output the singleton cluster $\{u\}$.} \COMMENT{Type $1$ cluster}
    \STATE{Let $S = S-\{u\}$.}
    \ELSE
    \STATE{Output the cluster $\{u\} \cup T$.} \COMMENT{Type $2$ cluster}
    \STATE{Let $S = S - (\{u\} \cup T)$.}
    \ENDIF
    \ENDWHILE
    \STATE{Output each remaining vertex of $V_2 \cap S$ as a singleton cluster.}
  \end{algorithmic}
\end{algorithm*}
\begin{theorem}\label{thm:bip-appx}
  Let $G$ be a labeled complete bipartite graph with partite sets
  $V_1$ and $V_2$, let $\alpha,\gamma$ be parameters as
  described in Algorithm~\ref{alg:bip-round}, and let $x$ be any
  fractional clustering of $G$. If $\cee$ is the clustering produced
  by Algorithm~\ref{alg:bip-round} with the given input, then for all
  $v \in V_1$ we have $\errvec(\cee)_v \leq c\errvec(x)_v$, where $c$
  is a constant depending only on $\alpha$ and $\gamma$.
\end{theorem}
We note that the proof of Theorem~\ref{thm:bip-appx} is actually simpler
than the proof of Theorem~\ref{thm:complete-appx}, because the focus on
errors only at $V_1$ eliminates the need for the ``bad pivots'' argument
used in Theorem~\ref{thm:bip-appx}. This also leads to a smaller
value of $c$ in Theorem~\ref{thm:bip-appx} than we were
able to obtain in Theorem~\ref{thm:complete-appx}.
\begin{proof}
  As before, we make charges to pay for the new cluster costs at each
  vertex of $V_1$ as each cluster is output, splitting into cases
  according to the type of cluster.  Let $k_1$ be a constant to be
  determined, with $k_1\alpha > \gamma$.
  
  \caze{1}{A Type 1 cluster $\{u\}$ is output.} In this case, the only
  cluster costs incurred are the positive edges incident to $u$, all
  of which have their other endpoint in $V_2$. The averaging argument
  used in Case~1 of Section~\ref{sec:algcc} shows that charging every
  edge incident to $u$ a factor of $2/\alpha$ times its LP cost pays
  for the cluster cost of all such edges.
  
  \caze{2}{A Type 2 cluster $\{u\} \cup T$ is output.}
  Negative edges within the cluster are easy to pay for: if $w_1w_2$ is a negative edge within the cluster,
  with $w_i \in V_i$, then we have
  \[ 1 - x_{w_1w_2} \geq 1 - x_{uw_1} - x_{uw_2} \geq 1-2\alpha, \]
  so we can pay for the cluster-cost of such an edge by charging it a factor of $1/(1 - 2\alpha)$ times its LP-cost.

  We still must pay for positive edges joining the cluster with the rest
  of $S$; we call such edges \emph{cross-edges}. Each such edge must be
  paid for at its endpoint in $V_1$.

  If $z \in V_1$ is a vertex outside the cluster, then a standard
  argument (see \apxref{apx:details}) shows that the cross-edges for $z$ can be paid
  for by charging each such edge a factor of $\max\{1/(1-2\alpha),
  2/\alpha)\}$ times its LP cost.

  Now let $w \in V_1$ be a vertex inside the cluster. We must pay for
  the cross-edges incident to $w$ using the LP-cost of the edges
  incident to $w$. First consider the positive edges from $w$ to
  vertices $z$ outside the cluster such that $x_{wz} \geq \gamma$. Any
  such edge has cluster-cost $1$ and LP-cost at least $\gamma$, so
  charging each such edge a factor of $1/\gamma$ times its LP-cost
  pays for its cluster cost. Let $X = \{z \in (S \cap V_2)-T \st
  x_{wz} < \gamma\}$; we must pay for the edges $wz$ with $z \in
  X$. Note that $x_{uz} > \alpha$ for all $z \in X$, since $z \in X$
  implies $z \notin T$.

  If $x_{uw} \leq k_1 \alpha$, then for all $z \in X$, we have
  \[ x_{wz} \geq x_{uz} - x_{uw} \geq (1-k_1)\alpha. \]
  Hence, for any positive cross-edge $wz$ with $z \in X$, the LP-cost of $wz$ is at least $(1-k_1)\alpha$,
  and so we can pay for the cluster-cost of $wz$ by charging $wz$ a factor of $\frac{1}{(1- k_1)\alpha}$
  times its LP-cost.

  Now suppose $x_{uw} > k_1\alpha$. As before, we pay for the cross-edges by charging
  the edges inside the cluster. Observe that $\sizeof{T^*_w} \geq
  \sizeof{X}$. Since $u$ was chosen to maximize $\sizeof{T^*_{u}}$, this
  implies that $\sizeof{T^*_u} \geq \sizeof{X}$. For any $v \in T^*_u$, we have
  \[ x_{wv} \geq x_{uw} - x_{uv} \geq k_1\alpha - \gamma. \]
  On the other hand, for any $v \in T^*_u$ we also have
  \[ 1 - x_{wv} \geq 1 - x_{uw} - x_{uv} \geq 1 -
  \alpha - \gamma \geq \alpha - \gamma. \]
  Since $k_1 \leq 1$, it follows that the edge $wv$ has LP-cost at least
  $k_1\alpha - \gamma$ independent of whether $wv$ is positive or
  negative. Thus, the total LP cost of edges joining $w$ to $T^*_u$ is
  at least $(k_1\alpha - \gamma)\sizeof{T^*_u}$.

  Since the total cluster-cost of the cross- edges joining $w$ and $X$
  is at most $\sizeof{X}$ and since $\sizeof{T^*_u} \geq \sizeof{X}$, we
  can pay for the cross-edges by charging each edge $wv$ with $v \in
  T^*_u$ a factor of $\frac{1}{k_1\alpha - \gamma}$ times its
  LP-cost.

\medskip Having paid for all cluster-costs, we now look at the total
charge accrued at each vertex. Fix a vertex $v \in V_1$ and an edge
$vw$ incident to $v$. We bound the total amount charged to $vw$ by $v$
in terms of the LP-cost of $vw$. There are three distinct
possibilities for the edge $vw$: either $vw$ ended inside a cluster,
or $v$ was clustered before $w$, or $w$ was clustered before $v$.

\caze{1}{$vw$ ended within a cluster.} In this case, $v$ may have made the
following charges:
\begin{itemize}
\item A charge of at most $\frac{1}{1 - 2\alpha}$ times the LP cost, to pay for $vw$ itself if $vw$ is a negative edge,
\item A charge of $\frac{1}{k_1\alpha - \gamma}$ times the LP-cost, to pay for positive edges leaving the $v$-cluster.
\end{itemize}

Thus, in this case the total cost charged to $vw$ by $v$ is at most $c_1$ times the LP-cost of $vw$, where
\[ c_1 = \frac{1}{1-2\alpha} + \frac{1}{k_1\alpha - \gamma}. \]

\caze{2}{$v$ was clustered before $w$.} In this case, $v$ may have made the following charges:
\begin{itemize}
\item A charge of $2/\alpha$ times the LP cost, to pay for $vw$ if $v$ was output as a singleton,
\item A charge of $\max\{\frac{1}{(1-k_1)\alpha}, \frac{1}{\gamma}\}$ times the LP cost, to pay for $vw$ if $v$ was output in a nonsingleton cluster,
\end{itemize}
Since $v$ makes at most one of the charges above, the total cost charged
to $vw$ by $v$ is at most $c_2$ times the LP-cost of $vw$, where
\[ c_2 = \max\left\{\frac{1}{(1-k_1)\alpha},\ \frac{1}{\gamma},\ \frac{2}{\alpha}\right\}. \]

\caze{3}{$w$ was clustered before $v$.} In this case, $v$ may have made the following charges:
\begin{itemize}
\item A charge of at most $\max\{\frac{1}{1-2\alpha}, \frac{2}{\alpha}\}$ times
  the LP cost, to pay for cross-edges at $v$ if $w$ is output in a nonsingleton cluster.
\end{itemize}
Thus, in this case the total cost charged to $vw$ by $v$ is at most $c_3$ times the LP-cost of $vw$, where
\[ c_3 = \max\left\{\frac{1}{1-2\alpha},\ \frac{2}{\alpha}\right\}. \]
The approximation ratio is $\max\{c_1, c_2, c_3\}$. Numerically, we obtain an approximation ratio
of at most $10$ by taking the following parameter values:
\[
  \alpha = 0.377 \qquad \gamma = 0.102 \qquad k_1 = 0.730 \qedhere
\]
\end{proof}
\section{Acknowledgments}
The authors thank Dimitris Papailiopoulos for helpful discussions
\ifuseappendix that led to the example in \apxref{apx:minimax}\fi.
The authors also acknowledge funding from the NSF grants IOS 1339388 and CCF 1527636, 1526875, 1117980.
Research of the first author was supported by the IC Postdoctoral Program.
\ifjournal
\bibliographystyle{amsplain}
\else
\bibliographystyle{icml2016}
\fi
\bibliography{cluster,cluster1}
\ifuseappendix
\appendix\onecolumn
\section{Minimax Clustering and the Failure of Pivoting Algorithms}\label{apx:minimax}
In this appendix, we consider \emph{minimax clustering}, which is the
special case of $f$-Correlation Clustering where $f(y) = \max_{v \in
  V(G)}y_v$. Thus, in minimax clustering, we seek to minimize the
number of errors at the \emph{worst vertex} in the
clustering. Equivalently, we are trying to minimize the
$\ell^{\infty}$-norm of the error vector, in contrast to classical
correlation clustering, where we are trying to minimize the
$\ell^1$-norm. 

Minimax clustering is a representative example of the difficulties
which arise in moving from classical correlation clustering to the
more general $f$-Correlation Clustering problem. We will show that
some techniques which work well for the classical correlation clustering
problem break down in the minimax context.

Ailon, Charikar, and Newman~\cite{ACN1,ACN2} gave a beautifully simple
randomized $3$-approximation algorithm for classical correlation
clustering on complete graphs. Their algorithm is shown in
Algorithm~\ref{alg:ccpivot}. Since our rounding clustering in
Section~\ref{sec:algcc} is based on the Charikar--Guruswami--Wirth
algorithm with a modified pivoting rule, it is natural to ask whether
a similar modification to the Ailon--Charikar--Newman algorithm also
yields a constant-factor approximation algorithm for minimax
clustering.
\begin{algorithm}
\caption{Ailon--Charikar--Newman algorithm~\cite{ACN1,ACN2}.}
\label{alg:ccpivot}  
\begin{algorithmic}
\STATE{Let $S = V(G)$.}
\WHILE{$S \neq \emptyset$}
\STATE{Pick $v \in S$ uniformly at random.}
\STATE{Let $T = (\{v\} \cup N^+(v)) \cap S$.}
\STATE{Output the cluster $T$.}
\STATE{Let $S = S-T$.}
\ENDWHILE
\end{algorithmic}
\end{algorithm}

Unfortunately, it seems that there are severe obstacles to modifying
the ACN algorithm in this manner. For any positive integer $t$, let
$M_t$ be a graph on $2t$ vertices consisting of $t$ pairwise disjoint
edges, and let $G_t$ be the labeling of $K_{2t}$ in which the edges of $M_t$
are labeled $-$ and all other edges are labeled $+$.

Clearly, if all vertices of $G_t$ are placed in the same cluster (the
``giant clustering''), then there is only $1$ error at each vertex of
$G_t$. We show that all other clusterings of $G_t$ have many more
errors at some vertex.
\begin{lemma}\label{lem:matcherror}
  If $\cee$ is a clustering of $G_t$ with more than $1$ cluster, then some
  vertex of $G_t$ has at least $t-1$ errors in $\cee$.
\end{lemma}
\begin{proof}
  Let $X$ be the smallest cluster in $\cee$. Since $\cee$ has at least $2$ clusters,
  we have $\sizeof{X} \leq t$. For any $v \in X$, there is at most one $w \notin X$
  such that $vw$ is a negative edge. Hence, each $v \in X$ has at least $t-1$ incident
  errors.
\end{proof}
By Lemma~\ref{lem:matcherror}, any constant-factor randomized
algorithm for minimax clustering must return the giant clustering for
$G_t$ with probability $1 - O(1/t)$. On the other hand, if we modify
Algorithm~\ref{alg:ccpivot} by changing the rule for choosing the pivot
vertex $v$, the resulting algorithm still cannot produce the giant
clustering. It is difficult to see how Algorithm~\ref{alg:ccpivot}
could sensibly be modified in order to return the giant clustering for
$G_t$ with high enough probability.

We now consider the behavior of Algorithm~\ref{alg:round} on the graph
$G_t$. While the minimax objective function is not linear in the
variables $x_{uv}$, we can still model the $f$-Fractional Correlation
Clustering problem using the linear program $\LP$ shown in
Figure~\ref{fig:LP}.
\begin{figure*}
  \[\begin{aligned}
    \text{minimize } M\text{, subject to:} \\ 
    x_{uv} &\leq x_{uz} + x_{zv} &\quad\text{(for all distinct $u,v,z$)}\\
    \sum_{w \in N^+(v)}x_{vw} + \sum_{w \in N^-(v)}(1-x_{vw}) &\leq M &\quad\text{(for all $v \in V(G)$)}\\
    0 \leq x_e &\leq 1 &\quad\text{(for all $e \in E(G)$)}\\
    M &\in \reals \\
  \end{aligned}
  \]  
  \caption{LP formulation $\LP$ of $f$-Fractional Correlation Clustering, where $f(y) = \max_{v \in V(G)}y_v$.}
  \label{fig:LP}
\end{figure*}
\begin{figure*}
  \[\begin{aligned}
    \text{maximize } \sum_{v \in V(G)}d^-(v)\pi_v \text{, subject to:} \\
    -\pi_u - \pi_v + \hat\sigma_{u,v} &\leq 0 &&\quad\text{(for all $uv \in E^+(G)$)}\\
    \pi_u + \pi_v + \hat\sigma_{u,v} &\leq 0 &&\quad\text{(for all $uv \in E^-(G)$)}\\
    \sum_v \pi_v &\leq 1 \\
    \pi_z, \sigma_{u,v} &\geq 0 &&\quad {\text{(for all $z \in V(G)$ and all $u,v \in E(G)$)}}.
  \end{aligned}
  \]    
  \caption{Dual of $\LP$.}
  \label{fig:dual}
\end{figure*}

Since the algorithm presented in Section~\ref{sec:algcc} yields a
constant-factor approximation algorithm for minimax clustering, and
since every clustering of $G_t$ other than the giant clustering has
$t-1$ errors at some vertex, it is necessary that our rounding
algorithm, applied to an optimal solution of $\LP$, returns the giant
clustering for all sufficiently large $t$. This follows immediately
from the following result.
\begin{proposition}\label{prop:matching-sol}
  Let $\LP$ be the linear program shown in Figure~\ref{fig:LP},
  as formulated for $G_t$. If $t \geq 3$, then the unique optimal
  solution to $\LP$ has $x_{uv} = 0$ for all $uv \in E(G)$.
\end{proposition}
\begin{proof}
  The dual program to $\LP$ is shown in Figure~\ref{fig:dual},
  with the following variables:
  \begin{itemize}
  \item For each $v \in V(G_t)$, a variable $\pi_v$ corresponding
    to the constraint $\sum_{w \in
      N^+(v)}x_{vw} + \sum_{w \in N^-(v)}(1-x_{vw}) \leq M$,
  \item For each ordered triple $(u,v,z)$ where $u,v,z$ are
    distinct vertices of $V(G_t)$, a variable $\sigma_{(u,v,z)}$
    corresponding to the constraint $x_{uv} \leq x_{uz} + x_{zv}$.
  \end{itemize}
  For convenience of notation, we also introduce the abbreviation
  $\hat\sigma_{u,v}$ to stand for $\sum_{z \in V(G) - \{u,v\}}(-\sigma_{u,v,z} -\sigma_{v,u,z} 
  + \sigma_{z,u,v} + \sigma_{z,v,u} + \sigma_{u,z,v} + \sigma_{v,z,u})$. Observe that there
  are exactly $2t-2$ choices of $z$ to sum over.

  Now we define a dual solution. Let $u'u''$ be an edge of the
  negative matching. Consider the dual
  solution defined below:
  \begin{align*}
    \pi_{u'} = \pi_{u''} &= 1/2, & \sigma_{u', u'', z} &= 1/(2t-2) \text{ for all $z \notin \{u', u''\}$}, \\
    \pi_{v} &= 0 \text{ for all $v \notin \{u', u''\}$}, & \sigma_{u,v,z} &= 0 \text{ if $(u,v) \neq (u', u'')$.}
  \end{align*}
  Clearly this solution has an objective value of $1$; we check that it is feasible for $t \geq 2$.
  If $uv$ is an edge containing neither of $\{u', u''\}$, then $\pi_u = \pi_v = 0$
  and $\hat\sigma_{u,v} = 0$, since every term of $\hat\sigma_{u,v}$ is $0$. The edge
  $u'u''$ is a negative edge with $\pi_{u'} = \pi_{u''} = 1/2$, and after eliminating all
  the zero terms, we have
  \[
    \hat\sigma_{u'u''} = \sum_{z \in V(G) - \{u,v\}}(-\sigma_{u',u'',z})
    = -\sum_{z \in V(G) - \{u,v\}}\frac{1}{2t-2}
    = -1.
    \]
  Thus, $\pi_{u'} + \pi_{u''} + \hat\sigma_{u',u''} \leq 0$, as
  required. Finally, if $uv$ is a positive edge with $u \in \{u',
  u''\}$, say if $u = u'$, then the only nonzero term of
  $\hat\sigma_{uv}$ is $\sigma_{u,u'',v}$, and we have $-\pi_u - \pi_v + \hat\sigma_{u,v} = -1/2 + 1/(2t-2) \leq 0$
  as required. The same argument holds if $u = u''$.

  Since this solution has an objective value of $1$, matching the
  primal objective when $x_{uv}=0$ everywhere, it is clearly
  optimal. Furthermore, if $t \geq 3$, then for each positive edge
  incident to $u'$ or $u''$, there is slack in the corresponding
  constraint of the dual problem.  By complementary slackness, this
  implies that in any optimal solution to $\LP$, we have $x_{u'v} =
  x_{u''v} = 0$ for all $v \in V(G) - \{u', u''\}$.  The triangle
  inequality constraints in $\LP$ then imply that in an optimal primal
  solution, $x_{uv}=0$ for all $uv \in E(G)$.
\end{proof}
\section{\maxagg{} for Classical and Minimax Clustering}\label{apx:maxagg}
In this paper, we have mainly focused on studying the \mindis{} formulation
of $f$-Correlation Clustering, where we seek to minimize an objective
function related to the clustering errors in a candidate solution, and where
a $c$-approximation algorithm is an algorithm whose total error weight is at most
$c$ times the optimal weight.

An alternative formulation to \mindis{} is \maxagg{}, where we instead
seek to \emph{maximize} some function related to the edges that are
\emph{not} errors. In classical correlation clustering, this means
that we want to maximize the number of edges which are correct. In
minimax clustering, we wish to maximize the number of correct edges
at the vertex with the fewest correct edges. In both cases, an optimal
solution to \mindis{} is also an optimal solution to \maxagg{}, but their
approximation properties differ.

In the classical case, there is a trivial $2$-approximation algorithm
for \maxagg{} on arbitrary graphs: we can simply choose the better of
clustering with all vertices in separate clusters and the clustering
with all vertices in the same cluster. All negative edges are correct
in the first clustering and all positive edges are correct in the
second clustering, so taking the better of the two yields a clustering
with at least half the edges correct, which is clearly at least half
the value of an optimal clustering. Less trivially, Bansal, Blum, and
Chawla~\cite{BBC1,BBC2} gave a PTAS for \maxagg{}, so that any
approximation ratio greater than $1$ is achievable. In contrast, the
best approximation ratio known for \mindis{} on arbitrary graphs has a
ratio of $\log n$.

It is natural to ask whether some algorithm can also be found to
approximate \maxagg{} in the minimax context. The trivial
$2$-approximation algorithm no longer works, since if $G$ both has
vertices of high positive degree and high negative degree, then each
of the ``extreme'' clusterings will cause a large number of errors at
some vertex. We have not been able to find any constant-factor
approximation algorithm for the \maxagg{} formulation of minimax
clustering, even with the additional assumption that $G$ is a labeled
complete graph.

We now construct a graph which seems to be a good example of the
difficulties in designing an algorithm for this problem. For any $n$,
let $G_n$ be the complete graph on $n+1$ vertices, and fix some vertex
$u^* \in V(G_n)$. All edges incident to $u^*$ are labeled $+$, while
all other edges are labeled $-$. Thus, $u^*$ has positive degree $n$,
while all other vertices have positive degree $1$.

It is clear that only one type of integer clustering could be optimal:
cluster $u^*$ with some number $t$ of the remaining vertices, and
cluster all other vertices as singletons. This yields $t$ correct
edges at $u^*$, $n-t+1$ correct edges at each vertex clustered with
$u^*$, and $n-1$ correct edges at each singleton vertex. Thus, the
optimal clustering has $\floor{(n+1)/2}$ correct edges at its worst
vertex.

The following result demonstrates why algorithms based on LP rounding
are likely to have trouble finding a good clustering of $G_n$ under
the \maxagg{} objective. We reuse the LP formulation of \mindis{}
shown in Figure~\ref{fig:LP}; this is valid because when we seek an
exact solution, minimizing $M$ in Figure~\ref{fig:LP} is equivalent to
maximizing $\sizeof{V(G)}-1-M$, the weight of the correct edges at the worst
vertex.
\begin{proposition}
  Let $\LP$ be the linear program shown in Figure~\ref{fig:LP},
  as formulated for $G_n$. If $n \geq 2$, then
  the unique optimal solution to $\LP$ has $x_{u^*v} = 1/3$ for
  all $v \neq u^*$ and $x_{vw} = 2/3$ for all $vw \in E(G_n-u^*)$.
\end{proposition}
\begin{proof}
  In the proposed solution, we have $M = n/3$. To show that this
  solution is optimal and unique, we construct a solution to the dual
  program shown in Figure~\ref{fig:dual}, as in the proof of
  Proposition~\ref{prop:matching-sol}. Consider the dual solution
  defined by
  \begin{align*}
    \pi_{u^*} &= 1 - \frac{n}{3(n-1)}, & \sigma_{v,w,u^*} &= \frac{1}{3(n-1)}\text{ for all $vw \in E(G_n-u^*)$} \\
    \pi_{v} &= \frac{1}{3(n-1)}\text{ for all $v \neq u^*$}, & \sigma_{v,w,z} &= 0 \text{ if $z \neq u^*$.}
  \end{align*}
  Since $d^-(u^*) = 0$ and $d^-(v) = n-1$ for all $v \neq u^*$,
  the objective value of this solution is $n/3$. Thus, if this
  solution is feasible, then it is optimal. 
  
  To see that this solution is feasible, we observe that for
  $v,w \neq u^*$, we have $\hat\sigma_{v,w} = -\sigma_{v,w,u^*} - \sigma_{w,v,u^*} = -(\pi_v + \pi_w)$, so that $\pi_v + \pi_w + \hat\sigma_{v,w} \leq 0$ for all
  negative edges $vw$, as needed. On the other hand, for $v \neq u^*$
  we have
  \[ \hat\sigma_{u^*,v} = \sum_{z \notin \{u^*,v\}}(\sigma_{v,z,u^*} + \sigma_{z,v,u^*}) = 2(n-1)\pi_v. \]
  Since $\pi_{u^*} = 1 - n\pi_v$, this implies that
  \[ -pi_{u^*} - \pi_{v} + \hat\sigma_{u^*,v} = -(1 - n\pi_v) - \pi_v + 2(n-1)\pi_v = (n-1)\pi_v - 1 + 2(n-1)\pi_v = 0, \]
  so that $-\pi_{u^*} - \pi_v + \hat\sigma_{u^*,v} \leq 0$ for all positive
  edges $u^*v$, as needed. Since also $\sum_v \pi_v = 1$, we see that
  the proposed dual solution is feasible, so the given primal solution
  is optimal.

  Now we argue that the given primal solution is the unique optimal
  solution. Let $x$ be any optimal primal solution. For each edge
  $vw \in E(G_n - u^*)$, the dual variable $\sigma_{v,w,u^*}$ is
  nonzero in the dual solution above, so by complementary slackness we
  have $x_{vw} = x_{u^*v} + x_{u^*w}$.  Furthermore, since each
  $\pi_v > 0$, each $v \neq u^*$ must have total error weight equal to
  $M$, again by complementary slackness. Therefore, for each $v \neq u^*$, we have
  \begin{align*}
  M = \sum_{w \in N^+(v)}x_{vw} + \sum_{w \in N^-(v)}(1-x_{vw})
  &= x_{u^*v} + \sum_{w \notin \{v,u^*\}}(1 - (x_{u^*v} + x_{u^*w}))\\
  &= (n-1) - (n-3)x_{u^*v} - \sum_{w \neq u^*}x_{u^*w}.
  \end{align*}
  This implies that $x_{u^*v} = x_{u^*w}$ for all $v \neq w$. Letting
  $p$ denote this common value, we have $M = (n-1) - (n-3)p - np
  = (n-1) - (2n-3)p$. On the other hand, since $\pi_{u^*} > 0$,
  we also have
  \[
  M = \sum_{w \in N^+(u^*)}x_{u^*w} = np.
  \]
  Thus, $(n-1) - (2n-3)p = np$, which implies that $p=1/3$.
  Hence, in any optimal solution we have $x_{u^*v} = 1/3$ for 
  all $v \neq u^*$ and $x_{vw} = 2/3$ for all $vw \in E(G_n - u^*)$,
  as desired.
\end{proof}
Thus, the only optimal solution to the natural LP rounding is highly
symmetric, but the natural symmetric clusterings of $G_n$ -- into
either all singletons or into one giant cluster -- both have at most
$1$ correct edge at the worst vertex, which is far short of the
optimum value of $\floor{n/2}$ correct edges. We note that this does
not pose a problem for the \mindis{} formulation: in a
$c$-approximation for \mindis{}, we only promise that the generated
clustering has at most $c\ceil{n/2}$ errors at its worst vertex, and
if $c > 2$, then any clustering at all meets this guarantee.
\section{NP-Completeness of Minimax Clustering on Complete Graphs}\label{apx:npc}
To show that minimax clustering is NP-hard on complete graphs, we use
a reduction from the Partition-into-Triangles problem, originally stated
in \cite{garey} and attributed to Schaefer.
\begin{quote}
  \textbf{Partition into Triangles}\\
  \textbf{Input:} A graph $G$ with $\sizeof{V(G)} = 3q$ for some integer $q$. \\
  \textbf{Question:} Is there a partition of $V(G)$ into $q$ sets $V_1, \ldots, V_q$
  such that each set $V_i$ induces a triangle in $G$?
\end{quote}
Specifically, we reduce from the $4$-regular case:
\begin{theorem}[van Rooij, van Kooten Niekerk, Bodlaender~\cite{partition}]
  Partition into Triangles on $4$-regular graphs is NP-complete.
\end{theorem}
(Although this is not explicitly stated in \cite{partition}, it
follows immediately from two of their results: that the problem is
NP-hard on graphs of maximum degree at most $4$, and that every
partition-into-triangles instance with maximum degree at most $4$ can
be transformed in polynomial time into an equivalent $4$-regular
instance.)

To prove that minimax clustering is NP-hard, we use the following reformulation,
which is more convenient for our purposes.
\begin{quote}
  \textbf{$t$-Perfect Clustering}\\
  \textbf{Input:} A labeled complete graph $G$ together with a tolerance $t_v \in \pints$
  for each $v \in V(G)$.\\
  \textbf{Question:} Does $G$ admit a $t$-perfect clustering, that is, a clustering
  such that each vertex $v$ has at most $t_v$ incident mistakes?
\end{quote}
Taking $\lambda_v = 1/t_v$, we see that $G$ has a $t$-perfect
clustering if and only if the minimax-clustering value of the
resulting weighted graph is at most $1$.

Our NP-completeness proof mimics the proof given by Bansal, Blum, and
Chawla for the classical correlation clustering problem. Let $G$ be a
$4$-regular graph on $n$ vertices, where $n \geq 7$, and let $G'$ be
the labeled complete graph on the same vertex set whose positive edges
are exactly the edges of $G$. Observe that $G$ has a partition into
triangles if and only if $G'$ has a clustering with all clusters of
size at most $3$ and exactly $2$ mistakes at each vertex. The idea is
to expand $G'$ into a larger labeled complete graph $H$ such that in
an optimal clustering of $H$, every cluster has at most three
$G'$-vertices.

We use essentially the same construction as Bansal--Blum--Chawla. Let
$H$ consist of $G'$, augmented as follows. For every $3$-set
$\{u,v,w\} \subset V(G')$, add to $H$ a clique $C_{uvw}$ with $7$
vertices. All edges within $C_{uvw}$ are positive, all edges from
$C_{uvw}$ to the vertices $\{u,v,w\}$ are positive, and all other
edges incident to $C_{uvw}$ are negative.

We assign the following tolerances: each original vertex $u \in G'$
has $t_u = 7({n-1\choose 2} - 1)+2$, and each added vertex $v \in H-G'$ has $t_v = 3$.
\begin{lemma}\label{lem:leq3}
  If $H$ has a $t$-perfect clustering $\cee$, then every cluster of
  $\cee$ contains at most three vertices of $G'$, and every cluster of $\cee$
  contains vertices from at most exactly one clique of $H-G'$.
\end{lemma}
\begin{proof}
  First suppose that $\cee$ has a cluster $X$ containing vertices from
  two different cliques of $H-G'$.  Let $v_1, v_2$ belong to the
  cliques $C_1, C_2$ respectively. If $\sizeof{X \cap C_1} > 3$, then
  $v_2$ has more than $3$ incident mistakes, which exceeds its
  tolerance. On the other hand, if $\sizeof{X \cap C_1} \leq 3$, then
  since $\sizeof{C_1} = 7$, we have $\sizeof{C_1 - X} \geq 4$,
  so $v_1$ has at least $4$ incident mistakes, which again exceeds
  its tolerance. Thus, if $\cee$ is $t$-perfect, then every cluster
  contains vertices from at most one clique.

  Now suppose that $\cee$ has a cluster $X$ that does not contain
  vertices from any clique of $H-G'$. Since clusters are nonempty, $X$
  contains a vertex $v \in V(G')$. Since $v$ has $7{n-1 \choose 2}$
  neighbors in $V(H-G')$ and is not clustered with any of them, $v$
  has at least $7{n-1 \choose 2}$ incident mistakes, which exceeds its
  tolerance of $7{n-1 \choose 2} - 5$.

  Finally, suppose that $\cee$ has some cluster $X$ with at least four $G'$-vertices.
  Since $X$ contains vertices from at most one clique of $H-G'$, there is
  some vertex $v \in V(G') \cap X$ does not have any positive neighbors
  in $X \cap V(H-G')$. Since $v$ has a total of $7{n-1 \choose 2}$ positive
  neighbors in $H-G'$, it again follows that $v$ has at least $7{n-1 \choose 2}$
  incident mistakes, exceeding its tolerance.
\end{proof}
\begin{corollary}
  $H$ has a $t$-perfect clustering if and only if $G$ has a partition
  into triangles.
\end{corollary}
\begin{proof}
  First suppose that $V_1, \ldots, V_k$ is a partition of  $G$ into triangles. Cluster $H$
  as follows: for $i \in [k]$, let $X_i = V_i \cup C_{V_i}$, where $C_{V_i}$ is the clque
  of $H$ with vertex set $V_i$. For every clique $C$ that is not equal to some $V_i$,
  cluster $C$ on its own.

  Each $v \in V(G')$ has exactly $7({n-1 \choose 2}-1) + 2$ mistakes:
  among the $7{n-1\choose 2}$ postive edges to vertices of $H-G'$, it
  is clustered with exactly $7$ of them, and among its $4$ positive
  neighbors in $G$, it is clustered with exactly $2$ of them (and with
  no negative neighbors), since $V_1, \ldots, V_k$ is a partition of
  $G$ into triangles. Furthermore, each $v \in V(H - G')$ has at most
  $3$ mistakes, since this clustering has no mistakes within $H-G'$
  and does not cluster any $w \in V(C_{xyz})$ with a vertex outside of
  $\{x,y,z\}$. Thus, the clustering is $t$-perfect.

  Now suppose that $H$ has a $t$-perfect clustering $\cee$. By
  Lemma~\ref{lem:leq3}, every cluster of $\cee$ contains at most three
  vertices of $G$ and contains vertices from exactly one cluster
  $C_{uvw}$ of $V(H-G')$.  We claim that the restriction of $\cee$ to
  $V(G')$ is a partition of $G$ into triangles. If not, some vertex $v
  \in V(G')$ is clustered with fewer than $2$ of its positive
  neighbors, and therefore has at least $3$ incident mistakes in
  $G'$. Since the cluster containing $v$ contains vertices from only
  one of the cliques containing $v$, we see that $v$ also has at least
  $7({n-1 \choose 2}-1)$ incident mistakes to vertices of $V(H'-G)$,
  for at total of at least $7({n-1\choose 2}-1)+3$ incident mistakes.
  This exceeds its tolerance, contradicting the hypothesis that $\cee$
  is $t$-perfect.
\end{proof}
\section{NP-Completeness on Complete Bipartite Graphs}\label{apx:bip-npc}
In this section, we show that ``one-sided'' minimax clustering on
complete bipartite graphs is NP-hard. This complements the
approximation algorithm given in Section~\ref{sec:algbc} for the same
problem. Our proof is similar to the proof of
Amit~\cite{amit2004bicluster} which shows that biclustering with the
classical objective function is NP-hard, but requires significant
modifications to accomodate the new objective function. The proof uses
a reduction from the \emph{$3$-cover problem}, which is well-known to be
NP-complete~\cite{garey}.
\begin{quote}
  \textbf{$3$-Cover}\\
  \textbf{Input:} A ground set $U = \{u_1, \ldots, u_{3n}\}$ and a family of
  subsets $\sey = \{S_1, \ldots, S_p\}$ with each $\sizeof{S_i} = 3$. \\
  \textbf{Question:} Is there a subfamily $\sey' \subset \sey$
  such that each $u_i$ lies in exactly one element of $\sey'$?
\end{quote}
Given an instance of $3$-cover, we construct an instance of the
following problem:
\begin{quote}
  \textbf{One-Sided $t$-perfect Biclustering}\\
  \textbf{Input:} A labeled complete bipartite graph $G$ with partite
  sets $V_1, V_2$ and a tolerance $t_v \in \pints$
  for each $v \in V_1$.\\
  \textbf{Question:} Does $G$ have a clustering such that each vertex
  $v \in V_1$ has at most $t_v$ incident edges that are errors?
\end{quote}
By the same argument used in Appendix~\ref{apx:npc}, any algorithm
which exactly determines the optimal one-sided minimax clustering for
complete bipartite graphs would also solve the $t$-perfect
biclustering problem. Hence, it suffices to show that $t$-perfect
biclustering is NP-hard. Note also that one-sided minimax clustering
can be viewed as the special case of (two-sided) minimax clustering
for which $t_v = \sizeof{V_1}$ for all $v \in V_2$; thus, the
reduction in this section also shows that the two-sided version of the
problem is NP-hard.

Given a nontrivial instance of $3$-cover (that is, an instance with
$n,p \geq 1$), we construct an instance of $t$-perfect biclustering as
follows. For each $u_i \in U$, construct a pair of vertices $x_i \in
V_1$, $y_i \in V_2$. Call these vertices \emph{ground vertices}.  Each
edge $x_iy_j$ is positive if $u_i=u_j$ or if $u_i$ and $u_j$ lie in
some common triplet of $\sey$, and negative otherwise.

For each $S_i \in \sey$, we create a vertex $x(S_i) \in V_1$
and $m$ vertices $y_1(S_i), \ldots, y_m(S_i) \in V_2$, where
each $x_j(S_i) \in V_1$ and $y_j(S_i) \in V_2$, where $m \geq 6n+3p$
is some fixed constant. Call these vertices \emph{triplet vertices},
and let $B_i = \{x(S_i)\} \cup \{y_j(S_i) \st j \in \{1, \ldots, m\}\}$.
All edges $x(S_i)y_k(S_i)$ for a fixed $i$ are positive, and all
edges $x(S_i)y_k(S_{\ell})$ for $i \neq \ell$ are negative. For
$u_i \in U$, if $u_i \in S_j$, then the edges $x_iy_k(S_j)$ and
$y_ix(S_j)$ are positive, and otherwise these edges are
negative.

Finally, let $Z = \{z_1, \ldots, z_{3n}\}$ be new $V_2$-vertices, and
for each $z_i \in Z$, add positive edges to all ground-vertices in
$V_1$ and negative edges to all triplet-vertices in $V_1$. Call these
vertices \emph{dummy vertices}.

Next we determine the tolerances $t_v$. For $S_i \in \sey$, let
$t_{x(S_i)} = 3$.
For $u_i \in U$, the corresponding tolerances are computed more
intricately. Let $d(u_i)$ be the number of triplets
$S_j \in \sey$ containing $u_i$ and let $c(u_i)$ be the number
of $u_j \in U - \{u_i\}$ such that $u_j$ and $u_i$ lie in some common
triplet $S_j$. We define
\[ t_{x_i} = m(d(u_i) - 1) + (c(u_i) - 2) + (\sizeof{Z} - 3). \]
It is clear that $G$ and $t$ can be constructed in polynomial time.
\begin{lemma}\label{lem:disjoint}
  Suppose that $G$ has a $t$-perfect clustering $\cee$. For any $S_i, S_j \in \sey$
  with $i \neq j$, the vertices $x(S_i)$ and $x(S_j)$ lie in different clusters.
\end{lemma}
\begin{proof}
  Suppose that $x(S_i)$ and $x(S_j)$ lie in the same cluster $X$. Since $t_{x(S_i)} = 3$,
  we see that $X$ contains at least $m-3$ vertices from $y_1(S_i), \ldots, y_m(S_i)$. Since
  $x(S_j)$ has negative edges to all these vertices, it follows that $x(S_j)$ has at least
  $m-3$ incident errors. Since $m-3 > 3 = t_{x(S_j)}$, this contradicts the fact that $\cee$
  is $t$-perfect.
\end{proof}
\begin{lemma}\label{lem:triplets}
  Suppose that $G$ has a $t$-perfect clustering $\cee$. For any $u_j \in U$,
  there is a unique $S_i \in \sey$ such that $x_j$ is clustered with $x(S_i)$.
  Furthermore, this $S_i$ has the following properties:
  \begin{enumerate}
  \item $u_j \in S_i$, and
  \item $x_j$ is clustered with each vertex $y_{\ell}$ such that $u_{\ell} \in S_i$.
  \end{enumerate}
\end{lemma}
\begin{proof}
  First we prove the existence of a unique $S_i$ such that $x_j$ is clustered
  with $x(S_i)$, then we show that $S_i$ has the desired properties.

  If $y_k(S_i)$ is a triplet $V_2$-vertex not clustered with $x(S_i)$, call $y_k(S_i)$
  a \emph{rogue vertex}. It is immediate from the definition of $t$ that in a $t$-perfect
  clustering, each $B_i$ contains at most $3$ rogue vertices.
  
  To prove that $x_j$ is clustered with some $x(S_i)$, it suffices to
  show that $x_j$ is clustered with some triplet $V_2$-vertex that is
  not a rogue vertex. Since each $B_i$ contains at most $3$ rogue
  vertices, there are at most $3p$ rogue vertices in total, where $p =
  \sizeof{\sey}$. If all triplet vertices clustered with $x_j$ are
  rogue vertices, then since $x_j$ has $md(u_j)$ positive edges to
  triplet vertices, it follows that $x_j$ has at least $md(u_j) - 3p$
  incident errors. Now we have
  \[ t_{x_j} = m(d(u_j) - 1)+ (c(u_j) - 2) + (\sizeof{Z} - 3) < md(u_j) - m + 6n \leq md(u_j) - 3p, \]
  where the last inequality follows from $m \geq 6n+3p$. Thus, there are more than $t_{x_j}$ errors
  at $x_j$, contradicting the assumption that $\cee$ is $t$-perfect. Thus, $x_j$ is clustered
  with some $x(S_i)$. Uniqueness of $S_i$ follows immediately from Lemma~\ref{lem:disjoint}.

  To see that $u_j \in S_i$, suppose that $u_j \notin S_i$. Then $x_j$
  is clustered with at most $3$ triplet-vertices that are its positive
  neighbors, and therefore has at least $md(u_j) - 3$ incident
  errors. Since $md(u_j)-3 > t_{x_j}$, this contradicts the assumption
  that $\cee$ is $t$-perfect.

  Next we prove (2). Let $B = N^+(x_j) - N^+(x(S_i))$. Since $t_{x(S_i)} = 3$, the cluster containing $x_j$
  contains at most $3$ vertices from $B$. Thus, there are at least $\sizeof{B}-3$ errors from $x$
  to the vertices of $B$, where
  \[ \sizeof{B} - 3 = \sizeof{Z} + m(d(u_j) - 1) + (c(u_j) - 2) - 3 =
  t_{x_j}. \] Thus, for $\cee$ to be $t$-perfect, it is necessary that
  \emph{all} errors incident to $x_j$ are edges from $x$ to $B$. In
  particular, $x_j$ is clustered with all vertices in $N^+(x_j) \cap
  N^+(x(S_i))$, so that $x_j$ is clustered with all $y_{\ell}$ such
  that $y_{\ell} \in S_i$.
\end{proof}
\begin{corollary}
  $G$ has a $t$-perfect clustering if and only if $\sey'$ has a $3$-cover.
\end{corollary}
\begin{proof}
  Given any $t$-perfect clustering, let $\sey'$ be the family
  of triplets $S_i$ such that some vertex of $B_i$ is clustered
  with some $V_1$-ground-vertex $x_j$. Lemma~\ref{lem:triplets} immediately
  implies that these triplets cover all of $u$. Furthemore, Lemma~\ref{lem:triplets}
  implies that these triplets are pairwise disjoint: if $S'_1$ and $S'_2$ are triplets
  of $\sey'$ that both contain $u_j$, then Lemma~\ref{lem:triplets} would force
  each $x(S'_1)$ and $x(S'_2)$ to both be clustered with $y_j$ and hence to be clustered
  together, which contradicts Lemma~\ref{lem:disjoint}. Hence, $\sey'$ is a $3$-cover.

  Conversely, let $\sey'$ be a $3$-cover in $\sey$. We define
  a clustering of $G$. Since $\sey'$ is a $3$-cover, we have $\sizeof{\sey'} = n$.
  Let $Z_{S'_1}, \ldots, Z_{S'_n}$ be a partition of $Z$ into $n$ disjoint sets of size $3$,
  indexed by the sets of $\sey'$. Now for each $S_i \in \sey$, define a cluster $X_i$ by
  \[ X_i =
  \begin{cases}
    B_i \cup \{x_j, y_j \st u_j \in S_i\} \cup Z_{S_i}, &\text{if $S_i \in \sey'$,} \\
    B_i, &\text{otherwise.}
  \end{cases} \] Since $\sey'$ is a $3$-cover, the clusters $X_i$ are
  pairwise disjoint and cover the vertices of $G$. We claim that this
  clustering is $t$-perfect.  If $x(S_i)$ is a triplet vertex
  corresponding to some $S_i \notin \sey'$, then $x(S_i)$ has exactly
  $3$ incident errors, namely its edges to the ground-vertices $y_j$
  with $u_j \in S_i$. On the other hand, if $x(S_i)$ is a triplet
  vertex corresponding to some $S_i \in \sey'$, then $x(S_i)$ again
  has exactly $3$ incident errors, namely its edges to the
  dummy-vertices in $Z_{S_i}$.

  If $x_j$ (or $y_j$) is a ground vertex, then $x_j$ has $m(d(u_j) -
  1)$ incident errors which are positive edges to triplet-vertices,
  $c(u_j) - 2$ incident errors which are positive edges to
  ground-vertices, and $\sizeof{Z}-3$ incident errors which are
  positive edges to dummy-vertices.  This is a total of exactly
  $t_{x_j}$ incident errors.  Hence the clustering is $t$-perfect.
\end{proof}
\section{Technical Details}\label{apx:details}
\begin{lemma}\label{lem:cross-edges}
  Suppose a Type~2 cluster $\{u\} \cup T$ has just been output in
  Algorithm~\ref{alg:round}. For any
  $z \in S - (\{u\} \cup T)$, the total cluster-cost of the
  cross-edges for $z$ is at most $\max\{1/(1-2\alpha), 2/\alpha\}$
  times the total LP-cost of the cross-edges for $z$.
\end{lemma}
\begin{proof}
  This is essentially the same proof given by Charikar, Guruswami, and Wirth~\cite{CGW1,CGW2};
  we repeat it here to keep the paper self-contained. If $x_{uz} \geq 1-\alpha$,
  then for each $w \in \{u\} \cup T$, we have
  \[ x_{wz} \geq x_{uz}-x_{uw} \geq 1-2\alpha. \]
  If there are $p$ positive cross-edges, this implies that the total LP-cost of the
  cross-edges for $z$ is at least $(1-2\alpha)p$. Since the total cluster-cost of
  the cross-edges for $z$ is $p$, the claim holds.

  Now consider $x_{uz} \in (\alpha, 1-\alpha)$. Let $P = N^+(z) \cap
  (\{u\} \cup T)$ and let $Q = N^-(z) \cap (\{u\} \cup T)$; the total
  cluster-cost of the cross-edges for $z$ is just $\sizeof{P}$. We
  have the following lower bound on the total LP-cost of the
  cross-edges for $z$:
  \begin{align*}
    \sum_{w \in P}x_{wz} + \sum_{w \in N}(1-x_{wz}) &\geq \sum_{w \in P}(x_{uz} - x_{uw}) + \sum_{w \in Q}(1 - x_{uz} - x_{uw}) \\
    &= \sizeof{P}x_{uz} + \sizeof{Q}(1-x_{uz}) - \sum_{w \in \{u\} \cup T}x_{uw} \\
    &\geq \sizeof{P}x_{uz} + \sizeof{Q}(1-x_{uz}) - \frac{\alpha(\sizeof{P} + \sizeof{Q})}{2},
  \end{align*}
  where in the last line we used the inequality $\sum_{w \in \{u\}\cup T}x_{uw} \leq \frac{\alpha\sizeof{\{u\} \cup T}}{2}$.
  This lower bound is linear in $x_{uz}$, so we study its behavior at the endpoints of $(\alpha, 1-\alpha)$.
  When $x_{uz}=\alpha$, the lower bound rearranges as follows:
  \[ \alpha\sizeof{P} + (1-\alpha)\sizeof{Q} - \frac{\alpha(\sizeof{P} + \sizeof{Q})}{2} = \frac{\alpha}{2}\sizeof{P} + (1-\frac{3\alpha}{2})\sizeof{Q} \geq \frac{\alpha}{2}\sizeof{P}.\]
  When $x_{uz}=1-\alpha$, the lower bound rearranges as follows:
  \[ (1-\alpha)\sizeof{P} + \alpha\sizeof{Q} - \frac{\alpha(\sizeof{P} + \sizeof{Q})}{2} = (1 - \frac{3\alpha}{2})\sizeof{P} + \frac{\alpha}{2}\sizeof{Q} \geq \frac{\alpha}{2}\sizeof{P}. \]
  In both cases, we used the assumption $\alpha < 1/2$, which implies $1 - \frac{3\alpha}{2} \geq \frac{\alpha}{2}$. It follows that charging $\frac{2}{\alpha}$ times the LP-cost of each cross-edge yields enough charge
  to pay for the cluster-cost of all cross-edges.
\end{proof}
\begin{lemma}\label{lem:bip-cross}
  Suppose that a Type 2 cluster $C$ has just been output in
  Algorithm~\ref{alg:bip-round}. For any vertex $z \in V_1 - C$, the
  total cluster-cost of the cross-edges for $z$ is at most
  $\max\{1/(1-2\alpha),\ 2/\alpha\}$ times the total LP-cost of the cross-edges for $z$.
\end{lemma}
\begin{proof}
  We essentially repeat the proof of Lemma~\ref{lem:cross-edges}. If $x_{uz} \geq 1-\alpha$,
  then for each $w \in \{u\} \cup T$, we have
  \[ x_{wz} \geq x_{uz}-x_{uw} \geq 1-2\alpha. \]
  If there are $p$ positive cross-edges, this implies that the total LP-cost of the
  cross-edges for $z$ is at least $(1-2\alpha)p$. Since the total cluster-cost of
  the cross-edges for $z$ is $p$, the claim holds.

  Now consider $x_{uz} \in (\alpha, 1-\alpha)$. Let $P = N^+(z) \cap
  (\{u\} \cup T)$ and let $Q = N^-(z) \cap (\{u\} \cup T)$; the total
  cluster-cost of the cross-edges for $z$ is just $\sizeof{P}$. Note
  that $P \cup Q = V_2 \cap T$. We have the following lower bound on
  the total LP-cost of the cross-edges for $z$:
  \begin{align*}
    \sum_{w \in P}x_{wz} + \sum_{w \in N}(1-x_{wz}) &\geq \sum_{w \in P}(x_{uz} - x_{uw}) + \sum_{w \in Q}(1 - x_{uz} - x_{uw}) \\
    &= \sizeof{P}x_{uz} + \sizeof{Q}(1-x_{uz}) - \sum_{w \in V_2 \cap T}x_{uw} \\
    &\geq \sizeof{P}x_{uz} + \sizeof{Q}(1-x_{uz}) - \frac{\alpha}{2}(\sizeof{P} + \sizeof{Q}),
  \end{align*}
  where in the last line we used the inequality $\sum_{w \in V_2 \cap T}x_{uw} \leq \frac{\alpha}{2}\sizeof{\{u\} \cup T}$.
  This lower bound is linear in $x_{uz}$, so we study its behavior at the endpoints of $(\alpha, 1-\alpha)$.
  When $x_{uz}=\alpha$, the lower bound rearranges as follows:
  \[ \alpha\sizeof{P} + (1-\alpha)\sizeof{Q} - \frac{\alpha}{2}(\sizeof{P} + \sizeof{Q}) = \frac{\alpha}{2}\sizeof{P} + (1-\frac{3\alpha}{2})\sizeof{Q} \geq \frac{\alpha}{2}\sizeof{P}.\]
  When $x_{uz}=1-\alpha$, the lower bound rearranges as follows:
  \[ (1-\alpha)\sizeof{P} + \alpha\sizeof{Q} - \frac{\alpha}{2}(\sizeof{P} + \sizeof{Q}) \geq (1 - \alpha - \frac{\alpha}{2})\sizeof{P} + \frac{\alpha}{2}\sizeof{Q} \geq \frac{\alpha}{2}\sizeof{P}. \]
  In both cases, we used the assumption $\alpha < 1/2$. It follows that when $x_{uz} \in (\alpha, 1-\alpha)$, charging $1/(\alpha-\beta)$ times the LP-cost of each cross-edge yields enough charge
  to pay for the cluster-cost of all cross-edges.
\end{proof}
\fi
\end{document}